%% file: pricing_the_unprovable.tex
\documentclass[runningheads,a4paper]{llncs}

\usepackage[T1]{fontenc}
\usepackage{amsmath,amssymb}
\usepackage{array}
\usepackage{booktabs}
\usepackage{graphicx}
\usepackage{tabularx}
\usepackage{xspace}
\usepackage{placeins}
\usepackage{xurl}
\usepackage{cite}
\usepackage[hidelinks,hypertexnames=false]{hyperref}
\ifdefined\pdfgentounicode
  \input{glyphtounicode}
\fi
\newcommand{\vdisc}{\textsc{VDISC}\xspace}

\title{Pricing the Unprovable: Ambiguity, Turnstile Runs, and the
Verifiability Discount in Crypto Protocols}
\titlerunning{Pricing the Unprovable: The Verifiability Discount}
\author{%
Doohwi~Cha\inst{1}\orcidID{0009-0000-2513-6965}\thanks{Corresponding author.} \and
Hanhee~Kim\inst{2}\orcidID{0009-0008-1280-7043} \and
Hyeonjun~Yeo\inst{3}\orcidID{0009-0005-9368-306X} \and
Minsuk~Sung\inst{4}\orcidID{0009-0007-2282-6873} \and
Minjae~Lee\inst{2}\orcidID{0000-0002-8937-9605} \and
Seunghan~Son\inst{2}\orcidID{0009-0000-2902-9999} \and
Donghwa~Seo\inst{5}\orcidID{0009-0003-5718-1535} \and
Juyeong~Lee\inst{6}\orcidID{0000-0002-4619-925X}}
\authorrunning{D. Cha et al.}
\institute{%
Mirae Asset Securities, Seoul, South Korea\\
\email{doohwicha@gmail.com}
\and
Independent Researcher, Seoul, South Korea\\
\email{crypt0xb0@gmail.com}, \email{abraxasnz13@gmail.com},
\email{qzpm5n@g.skku.edu}
\and
Dept. of Physics and Astronomy, Seoul National University, Seoul, South Korea\\
\email{hj.yeo30@gmail.com}
\and
Department of Artificial Intelligence, Korea University, Seoul, South Korea\\
\email{minsuksung@korea.ac.kr}
\and
DS Investment \& Securities, Seoul, South Korea\\
\email{properitas95@kaist.ac.kr}
\and
EY Consulting, Seoul, South Korea\\
\email{juyeong.lee0839@gmail.com}}

\begin{document}
\raggedbottom
\maketitle

\begin{abstract}
Privacy-preserving protocols can repair a vulnerability without establishing
whether it was exploited beforehand, even when no theft is reported and
aggregate supply is protected. We call the potential value reduction from this
gap the \emph{verifiability discount}. A turnstile model shows how early
withdrawal can shift a possible shortfall to remaining holders and why credible
evidence can become more valuable as exposure grows.
We introduce \vdisc, a reproducible framework that records reported loss,
exposed value, public evidence, returns, and pool flows separately. It combines
Orchard's staged disclosures with reconstructed on-chain flows, earlier Zcash
cases, and a 40-event security panel.
Orchard follows the predicted direction: market responses were muted after
repair and supply assurance, while the largest negative return and aggregate shielded-pool net outflow
aligned with the acknowledgment that past abuse could not be ruled out. These
patterns are directional, not causal. Across events, market-cap-normalized
reported loss and unresolved exposure overlap among the 38 events with a
public loss figure (Pearson $0.86$, Spearman $0.56$), partly because exposure is
often coded from the loss itself, so returns alone cannot distinguish reported
loss from unresolved history.
The findings motivate three candidate design priorities: produce
privacy-compatible evidence about past states, limit the state exposed by a
single flaw, and disclose residual uncertainty explicitly.

\end{abstract}

\section{Introduction}

When a security flaw is disclosed, investors ask how much was lost and whether
the flaw was fixed. Privacy-preserving protocols add a third question: can
outsiders verify what happened before the fix? A patch may stop new abuse and a
supply rule may protect total supply, yet private transaction history can leave
past abuse impossible to rule out. Finding no abuse is therefore weaker than
proving that abuse could not have occurred.

The June 2026 Zcash Orchard sequence separated these questions in time. The
repair announcement reported no evidence of exploitation or unauthorized value
creation and no detected impact on total supply \cite{zodl2026remediated};
hours later the Zcash Foundation confirmed both, citing the turnstile, which
limits transfers out of each shielded pool to its recorded balance
\cite{zcashfoundation2026nu62}. The disclosing team later judged prior
exploitation unlikely but stated that public evidence could not rule it out
\cite{wilcox2026orchard}. ZEC was already falling before that statement, so
timing alone is not causal; the sequence instead permits a staged comparison of
early assurance---notice, repair, and supply protection---with later unresolved
history.

We call the value reduction associated with adverse histories that remain
possible after repair the \emph{verifiability discount}. Investors need not
agree on a precise hidden-loss probability: leaving several adverse histories
open can itself affect valuation and exit. This is ambiguity in the sense of
decision-making under uncertain probabilities
\cite{gilboa1989maxmin,klibanoff2005smooth,epstein2008ambiguity}.

The paper makes five contributions. First, it defines a pricing construct that
separates reported loss from adverse histories that public evidence cannot
close. Second, it develops a formal turnstile model. Early withdrawal can shift
more of a possible shortfall to remaining holders; locally, exit pressure rises
with unresolved exposure (Proposition~\ref{prop:p4}), and proof can be more
valuable when more value was exposed (Proposition~\ref{prop:p5}). Third, it
builds a four-anchor Orchard evidence design that pairs returns with net
unshielding---ZEC leaving the aggregate shielded pools minus ZEC entering them.
Fourth, earlier Zcash cases, 14 coded disclosures, and a 40-event panel make the
identification boundary explicit: Orchard and the earlier Zcash cases show
timing and direction rather than a causal effect, the cross-event pilot does
not identify even the sign of the CAR--$z$ slope, and reported loss and
unresolved history require separate measures. Fifth, \vdisc records source and time, reported loss, exposure, public
proof, returns, and pool flows separately and maps them to three design
priorities: strengthen past-state evidence, limit single-flaw exposure, and
state what remains unresolved. Figure~\ref{fig:framework} summarizes the
mechanism.

\begin{figure}[t]
\centering
\includegraphics[width=\textwidth]{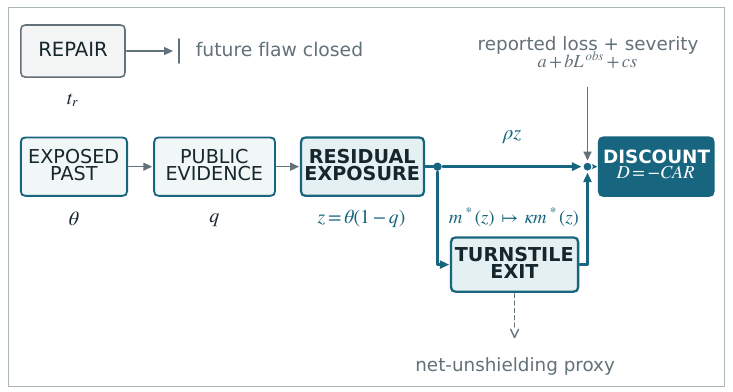}
\caption{Evidence pipeline for the verifiability discount. Repair at $t_r$
closes future abuse but does not resolve the past. Public evidence $q$ maps
exposure $\theta$ to residual $z=\theta(1-q)$. On a locally stable branch,
$z$ enters $D=-CAR$ directly through $\rho z$ and indirectly through turnstile
exit $m^*(z)$ and price pressure $\kappa m^*(z)$; reported loss $L^{obs}$ and
severity $s$ enter separately. Arrows encode model channels, not causal
estimates; net unshielding indicates exit direction but does not measure gross
exit.}
\label{fig:framework}
\end{figure}

\section{Related Literature}

Ambiguity arises when investors do not know which probability or model applies
\cite{gilboa1989maxmin,klibanoff2005smooth,epstein2008ambiguity}; markets price
it separately from volatility \cite{brenner2018asset}, including in crypto
returns \cite{luo2021fintech}. Security-disclosure and crypto event studies
show that news about loss, exposure, and controls moves prices, including after
51\% and broader attacks
\cite{campbell2003economic,cavusoglu2004effect,telang2007impact,spanos2016impact,shanaev2020cryptocurrency,ramos2021disturbance,cheraghali2024impact}. Orchard
adds a distinct object: a patched flaw and protected supply whose private past
still cannot be verified. The verifiability discount therefore records
reported loss separately from adverse histories that remain open.

Run models show why sequential withdrawals can induce exit before uncertainty
is resolved \cite{diamond1983bank,goldstein2005demand,morris1998unique}; related
incentives appear in stablecoins \cite{ahmed2024public,davernas2022stablecoins},
and concentrated selling depresses prices in shallow markets
\cite{brunnermeier2009liquidity}. Here early exit can shift more of a hidden
shortfall to remaining holders, linking incomplete proof to flow and price.

Zerocash introduced anonymous zero-knowledge payments
\cite{sasson2014zerocash}, whose soundness depends on correct circuit
computation \cite{gennaro2013quadratic}. In Zcash, transfers inside a shielded
pool are hidden while entries and exits remain visible
\cite{kappos2018empirical}. A later circuit repair may thus
leave earlier private use unresolved. Our model lets proof close possible
histories and the turnstile allocate any remaining shortfall.

\section{From Unverifiable History to Valuation and Exit}\label{sec:model}

The model asks what changes when future abuse is stopped but pre-repair abuse
cannot be ruled out. If some claims may be invalid, early exit avoids part of
that risk and spreads the same possible shortfall over fewer remaining claims.
Public proof narrows those possible histories. The model links the unresolved
remainder to a price discount and exit pressure.

For event $i$, exposure $\theta_i$ is the share of the relevant claim base that
was vulnerable to the disclosed flaw.
Public proof strength is $q_i\in[0,1]$. At $q_i=1$, public evidence rules out the
relevant adverse histories; at $q_i=0$, it rules out none of them. Residual
unverifiable exposure is therefore $z_i=\theta_i(1-q_i)$. Reported loss
$L_i^{obs}$ is the amount quantified by public sources, whereas actual loss in
private state, $L_i^*$, may be different or impossible to observe. The severity
signal $s_i$ records other observable event characteristics and deliberately
excludes residual unverifiability. Because $L_i^*$ is unobserved, it does not
enter the empirical equation directly:
\begin{equation}\label{eq:reducedform}
D_i = -CAR_i = a + bL_i^{obs} + cs_i + \gamma\,\theta_i(1-q_i) + \eta_i .
\end{equation}
Here a larger $D_i=-CAR_i$ means a larger negative abnormal return. The equation
records reported loss, observable severity, and residual unverifiable exposure
separately; $\gamma$ is the association with the last component after the first
two are included. A single return still cannot identify all three channels, so
the empirical analysis also examines on-chain exit behavior.

\subsection{A Two-Date Turnstile Run}

\paragraph{Choices and information.} At date~0, a vulnerability is disclosed
and patched; at date~1, a migration or supply proof either resolves the pool
history or leaves it open. Normalize holders, claims, and the value of each
valid claim to one. Invalid claims $x$ are unobserved, with
$x\in[0,\theta(1-q)]$. Exposure $\theta$ widens this range, proof $q$ narrows
it, and $z\equiv\theta(1-q)$ is the maximum unresolved claim value. Under
maxmin preferences \cite{gilboa1989maxmin}, investors evaluate the worst
publicly admissible case, $x=z$, rather than assigning probabilities to each
$x$.

Each holder exits or stays. With exit mass $m$, a slack turnstile processes
unshielding at par, while selling realizes $1-\kappa m$, where
$\kappa\in(0,1)$ measures price pressure. Exit also costs holder $j$ an idiosyncratic
$\varphi_j$ (tax, lost privacy, inattention, or other friction), distributed by
$F$ on $[0,\bar\varphi]$. We assume $F$ is continuously differentiable and
strictly increasing, with $f=F'>0$ near an interior cutoff, $F(0)=0$, and
$\bar\varphi>1-\kappa$.

\paragraph{Why exit can reinforce itself.} The turnstile limits unshielding to
the recorded pool balance and serves requests in block order, a sequential
service constraint \cite{diamond1983bank}. Invalid claims add demand without
capacity, so total demand can be $m+x$ against unit capacity. Even while the
turnstile is slack ($m+x\le1$), after masses $m$ and
$x$ leave, the remaining $1-m$ claims have only $1-m-x$ backing. At the worst
case $x=z$, a stayer receives $\max\{0,\,1-z/(1-m)\}$. More exit therefore
concentrates the possible shortfall on fewer stayers and can induce further
exit. This matches Zcash's emphasis on pool totals and remaining supply
\cite{electriccoin2019zcash}: unresolved history creates the incentive, and
the turnstile allocates the possible shortfall.

\paragraph{Equilibrium.} With a slack turnstile ($m\le1-z$), exit avoids
worst-case dilution but incurs individual cost and market impact. Holder $j$
exits when $\varphi_j<g(m;z)$, where
\begin{equation}\label{eq:cutoff}
g(m;z)\;=\;\frac{z}{1-m}\;-\;\kappa m .
\end{equation}
The first term is the dilution avoided by leaving; the second is the price cost
of joining a larger exit. The best-response mass is
$B(m;z)=F(\max\{0,g(m;z)\})$, and equilibrium requires $B(m;z)=m$. Because $m$
raises both avoided dilution and market impact, $B$ need not be monotone.

We focus on runs before capacity is exhausted, isolating shortfall
concentration from a specific rationing rule. For some $\hat z\in(0,1)$, impose
$z\le\hat z$ and $F(1-\kappa(1-\hat z))\le1-\hat z$; such a $\hat z$ exists
because $F(1-\kappa)<1$. Then $m>1-\hat z$ cannot be an equilibrium, since even
the most favorable rationing rule gives
$B(m;z)\le F(1-\kappa(1-\hat z))\le1-\hat z<m$. Continuity supplies a fixed
point in the slack region.

Because concentration and market impact move incentives in opposite
directions, we claim neither global monotonicity nor uniqueness. Comparative
statics apply only to differentiable, locally stable branches with
$|B_m(m^*;z)|<1$, where small deviations are not amplified.

On a locally stable branch, unresolved history affects the date-0 discount in
two ways. Investors apply a direct haircut to the possible shortfall, and the
resulting exits add market-price pressure:
\begin{equation}\label{eq:discount}
D \;=\; a + bL^{obs} + cs \;+\; \rho z \;+\; \kappa\,m^*(z).
\end{equation}
Here $\rho z$ is the direct valuation effect and $\kappa m^*(z)$ is exit-driven
price pressure. Thus the reduced-form $\gamma$ in
Eq.~\eqref{eq:reducedform} combines both responses locally. Returns contain
both; on-chain flow provides a separate directional check on exit.

The next two propositions formalize the model's central claims. The first
establishes when unresolved history produces exit pressure. The second shows
why proof is more valuable when a larger claim base was exposed.

\begin{proposition}[run intensity]\label{prop:p4}
Under the slack-turnstile condition, for $z=\theta(1-q)\in[0,\hat z]$: (i) if
$z=0$ (i.e.\ $q=1$ or $\theta=0$), the unique equilibrium is $m^*=0$; (ii) if
$z>0$, every equilibrium has $m^*>0$ and leaves the turnstile slack
($m^*+z\le1$) under any rationing rule; (iii) along any differentiable locally
stable threshold branch with an interior cutoff, $\mathrm dm^*/\mathrm dz>0$.
The empirical design uses only the corner at $q=1$ and this local directional
implication; it does not rely on global uniqueness or global monotonicity.
\end{proposition}

\begin{proposition}[proof--exposure complementarity]\label{prop:p5}
The direct ambiguity component of \eqref{eq:discount} has
$\partial^2[\rho\theta(1-q)]/\partial\theta\,\partial q=-\rho<0$: the price value
of proof strength scales with exposed value. The run component reinforces this
complementarity on differentiable stable branches whenever
$\mathrm d[z\,m^{*\prime}(z)]/\mathrm dz\ge0$; this curvature condition holds in
the uniform-cost benchmark but is not imposed globally.
\end{proposition}

\paragraph{Proof Sketch.} If $z=0$, then $g(m;0)\le0$, so
$B(\cdot;0)=0$ and $m^*=0$ is unique. If $z>0$, then
$B(0;z)=F(z)>0$, so zero cannot be an equilibrium; the slack condition excludes
capacity-binding equilibria. On a differentiable stable branch,
\[
\frac{\mathrm dm^*}{\mathrm dz}
=\frac{B_z(m^*;z)}{1-B_m(m^*;z)}>0.
\]
because $B_z=f(g)/(1-m)>0$ at an interior cutoff and local stability gives
$1-B_m>0$. Hence $m^*$ rises locally with $z$. The negative cross-partial in
Proposition~\ref{prop:p5} says that proof reduces the direct discount more when
exposure is larger; the run component shares this complementarity only under
the stated curvature condition.

\paragraph{Uniform-Cost Benchmark.} For illustration, let
$F(\varphi)=\varphi/\bar\varphi$ and $A=\bar\varphi+\kappa$, combining exit-cost
dispersion and market impact. If $0<z<A/4$, $\kappa<\bar\varphi$, and
$m_-^*(z)+z<1$, the locally stable branch is
\[
m_-^*(z)=\frac{1-\sqrt{1-4z/A}}{2},\qquad
\frac{\partial m_-^*}{\partial z}
=\frac{1}{A\sqrt{1-4z/A}}>0.
\]
The lower root is the stable slack solution. Its positive, increasing derivative
shows a locally steeper exit response as unresolved exposure grows. Magnitudes
depend on $A$, so this is intuition, not calibration.

\paragraph{Model Implications and Scope.} Complete proof or zero exposure sets
$z=0$ and removes ambiguity-driven exit; on a locally stable interior branch,
more $z$ raises exit. Stronger proof also reduces the direct discount more when
exposure is larger, with run reinforcement only under
Proposition~\ref{prop:p5}'s curvature condition. These are local signs under
maxmin evaluation, a slack turnstile, linear market impact, and local stability,
not claims of a run after every disclosure or of a unique global equilibrium.
The case study tests these directions rather than structural parameters. Its
net pool flow (outflow minus inflow) is only a directional proxy for the model's
gross exit mass $m$; prices, pool flow, and order flow are complementary rather
than interchangeable evidence.

Two distinctions matter when taking the model to data. First, $z=0$ is a sharp
boundary, whereas the positive-$z$ result is local and does not imply that every
incident produces a run. Second, observed net flow subtracts inflows from gross
outflows, while $m$ counts gross exit. Agreement in direction is informative,
but their magnitudes are not interchangeable. Transaction-level records,
protocol timestamps, or cryptographic commitments would map withdrawals to the
model more tightly and support nonlinear or adaptive extensions.

\subsection{Predictions and Their Empirical Tests}

The price equation gives \textbf{P1}: with weak proof, larger exposure
$\theta_i$ is associated with a larger discount; and \textbf{P2}: holding
exposure fixed, credible supply, containment, or migration evidence raises
$q_i$ and reduces the discount. These are directional, not causal, associations.

The turnstile gives \textbf{P4}: ambiguity-driven exit is zero at
$z=\theta(1-q)=0$ and rises locally with positive $z$ on a stable interior
branch. Net unshielding after unresolved-history news provides a directional
check, not an estimate of gross exit mass $m$. \textbf{P5} says the direct value
of proof rises with exposure; run reinforcement requires the stated curvature.
\textbf{P3}, an untested extension in which AI-assisted discovery changes
perceived vulnerability arrival, lies outside the two-date model.

Accordingly, four ZEC anchors examine P2, on-chain flow checks P4, and the
cross-event panel examines P1. P5 informs design but is not estimated. Returns
and flows test predicted directions, not a causal ambiguity effect.

\section{Measuring Loss, Unverifiability, and Holder Response with \vdisc}\label{sec:data}

Testing requires separate measures of damage, unresolved history, and holder
response. The coverage layer links 151 of 544 DeFiLlama seed incidents to token
prices, defining availability rather than an event-study sample
\cite{defillamahacks}. The analysis layer contains four official Bitcoin or
Zcash cases and ten disclosures coded by fixed rules; a separate 40-event panel examines
loss--unverifiability overlap. Included events require a timestamped disclosure
or remediation, a traded token, price history, and sufficient public records.

\vdisc records exposure $\theta$ and proof $q$ separately: a severe event need
not be hard to verify. Exposure is the potentially affected claim base, proxied
within failure class by shielded-pool balance, bridge or affected total value
locked (TVL), or excess minted value, each scaled by supply or market value;
ratios above one are capped and flagged. When a record quantifies only the
stolen amount, that amount serves as the affected value. Proof asks how much public evidence
closes adverse histories. Bounds on abuse, stolen funds, affected addresses, or
contaminated state strengthen $q$; material unresolved abuse weakens it.

Empirical $q$ averages five yes/no checks within two days: supply proof,
containment, replayability, migration proof, and an official residual-uncertainty
statement. Each requires a public URL and quote; later evidence is stored
separately. Because this window overlaps part of the daily cumulative abnormal
return (CAR), $q$ organizes evidence rather than acting as a predetermined
treatment. It also proxies auditability, not structural $q$: stating an unknown
makes it visible without necessarily reducing it.

Automated models retrieve candidate sources and propose labels, but those labels
do not construct $\theta$ or $q$ and do not enter the empirical tests. In a
frozen 50-event audit of GPT-5.4 Mini and Qwen3.7 Plus, 99 of 100 outputs
cited a source with an evidence summary; inter-model Cohen's kappa was $0.65$ for type and $0.58$ for severity but only
$0.002$--$0.052$ for judgment-dependent fields. Thus fixed formulas and
checklists construct $\theta$ and $q$; model-assisted judgment remains future
work.

For reproducibility and safety, collection uses public timestamped sources and
excludes exploit instructions, payloads, keys, discovery prompts, live-system
procedures, and other unsafe fields. Fixed fields and observation times preserve
Orchard's disclosure sequence instead of collapsing it to one date.

\section{Orchard: Four Anchors and One Unresolved Past}\label{sec:zec}

Repair and supply assurance arrived together; proof of past non-abuse did not
arrive in the sample window: the later statement said Orchard data could not
provide it cryptographically and proposed an upgrade to prove supply integrity.
We use four timestamped public statements as \emph{anchors}: the first notice,
the NU6.2 repair with its remediation statement, the Foundation's confirmation,
and the later statement that prior exploitation could not be ruled out.

After May 29 discovery, a soft fork disabled Orchard at block 3,363,426,
timestamped 03:30 UTC on June 2 \cite{blockchairapi} (official posts give
02:00--02:30; we use the block time), before the 06:20 notice
\cite{zcashfoundation2026coordinated}; it is an operational reference, not an
anchor. NU6.2 restored operation on June 3 at 04:05, and the 04:09 remediation
statement reported no evidence of exploitation or unauthorized value creation
and no detected impact on total supply \cite{zodl2026remediated}. At 10:55 the
Foundation confirmed both points \cite{zcashfoundation2026nu62}. On June 4 at
21:33:53, the context post judged prior exploitation unlikely but said Orchard
data could not rule it out cryptographically, and proposed a supply proof
\cite{wilcox2026orchard}. Hayes later announced an exit, but its execution time
is unknown \cite{hayes2026x}.

\subsection{Market Ordering and Event-Study Design}

The market test asks whether ZEC responds to each anchor's information. A
standard event study \cite{mackinlay1997event} subtracts its usual BTC-linked
return. For token $i$ and period $t$,
$R_{i,t}=\alpha_i+\beta_i R_{m,t}+\varepsilon_{i,t}$,
$AR_{i,t}=R_{i,t}-\hat\alpha_i-\hat\beta_i R_{m,t}$, and
$CAR_i(\tau_1,\tau_2)=\sum_{t=\tau_1}^{\tau_2}AR_{i,t}$. Thus CAR accumulates
ZEC return unexplained by hourly Binance BTC/USDT \cite{binancespotklines}.
The model uses 335 pre-discovery hours ending May 28 at 23:00 UTC; matched
pre/post windows contain 24 complete returns and omit the anchor bar.

\input{tables/table_zec_case}

Table~\ref{tab:zecdiag} shows the ordering: CAR is $0.185$ after notice,
$0.014$ and $0.016$ after the repair statement and its confirmation, and
$-0.162$ only after past exploitation cannot be ruled out. This matches the
predicted direction but is not causal identification.

The context tail probability is $p_{ob}=0.003$ and
$q_{\mathrm{BH}}=0.051$ across the 16-test family
\cite{benjamini1995controlling}; its pre-anchor CAR is already $-0.134$, which
reaches the same placebo floor. Corrado ranks
\cite{corrado1989nonparametric}, a GARCH volatility diagnostic motivated by
event-induced-variance concerns \cite{boehmer1991event}, log returns, and 60/90-day windows
($p\leq0.001$) preserve the negative sign; no leakage appears from May 30 00:00
to June 2 07:00. These checks do not
create a causal break: the final post bundled severity with unverifiability and
repricing had begun earlier.

The context window is therefore unusual, but it clears a conventional 5\%
threshold only before adjustment for multiple testing. Its negative pre-CAR is
also economically important: robustness to clocks and estimators cannot turn a
price decline already under way into a disclosure-time discontinuity.

\subsection{Minute Anchor Check and Cross-Venue Robustness}

The same statement first appeared as an X article at 21:17 UTC, 16 minutes
before the forum post \cite{zooko2026x}. Reanchoring to that time changes minute
CAR from $-0.166$ to $-0.205$ but not its sign. Same-window CARs on Coinbase, Kraken, OKX, and Bybit range from
$-0.155$ to $-0.166$ \cite{crossvenuecandles}; Bybit is a perpetual, so these
are not four independent spot replications.

Figure~\ref{fig:market} shows why this robustness is not a sharp break. CAR had
already fallen by $0.208$ since midnight ($0.188$ by 21:17). Turnover was
comparable to the June 3 crisis day until about the time of the X release and
about four times higher in the 16 minutes before the forum post; the
post-anchor 60-minute decline was about 1.6 times the pre-anchor decline. Hayes confirms an exit announcement, not execution time
\cite{hayes2026x}. The sign survives timestamp choices, but the earlier movement
precludes a discrete causal interpretation.

\begin{figure}[ht]
\centering
\includegraphics[width=\textwidth]{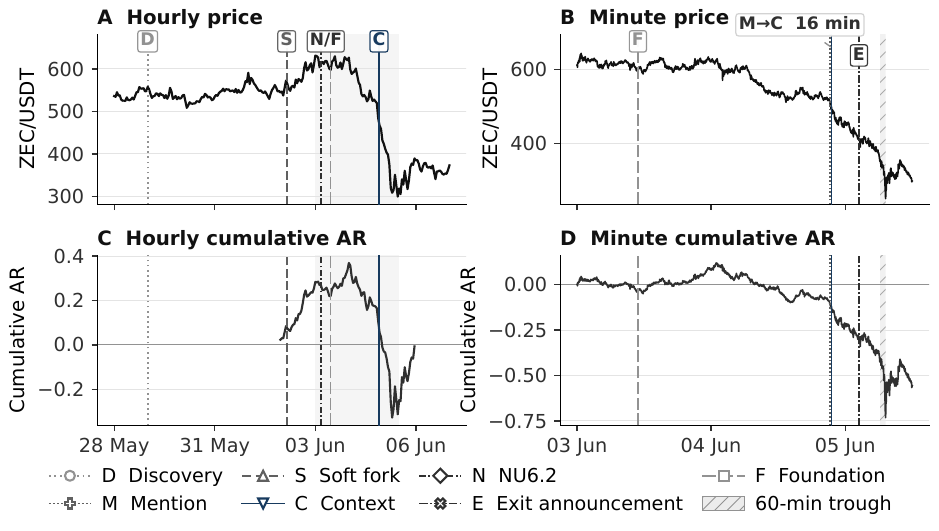}
\caption{Orchard market evidence. Panels A/C show hourly price and BTC-adjusted
CAR; B/D show minute data through June 5 12:00 UTC. Shading marks the zoomed interval
and largest 60-minute decline. N/F shares a top-rail tag because its two true
lines are only 6 h 50 min apart; M$\rightarrow$C marks the true 16-minute gap
between the 21:17 X article and the forum post.
Negative movement predates the final anchor, so markers establish timing rather
than causal breaks.}
\label{fig:market}
\end{figure}

\subsection{On-Chain Exit Evidence}

P4 is checked with \emph{net unshielding}: ZEC leaving the aggregate Sprout,
Sapling, and Orchard shielded pools minus entering them. This measures net
aggregate shrinkage rather than gross exits. Per-block changes
reconstruct June 1--7 daily/hourly flow \cite{blockchairapi}; levels agree within
0.1\% with zcashinfo/node snapshots \cite{zcashinfoapi,zcashexplorerapi}.
Pool-specific snapshots show that Orchard accounted for 97.4\% of the May
24--June 7 aggregate shielded decline.

\input{tables/table_anchor_flows}

In Figure~\ref{fig:pooltimeseries}, flow is small or constrained around pause/repair, but reaches 175,568 ZEC in the
24 hours after the context post, about $49\times$ baseline
(Table~\ref{tab:anchorflows}). From the first complete 22:00 UTC bin, 48 hours
total 243.2k ZEC, or
4.8\% of the aggregate shielded balance. A single block (June 5, 19:34 UTC; 20 transactions)
nets 60,324 ZEC of outflow, 34\% of the 24-hour total; without it, the 24-hour
outflow is still about $32\times$ baseline. Thus the largest net outflow aligns with the widest publicly
recognized unresolved history.

\FloatBarrier
\begin{figure}[ht]
\centering
\includegraphics[width=\textwidth]{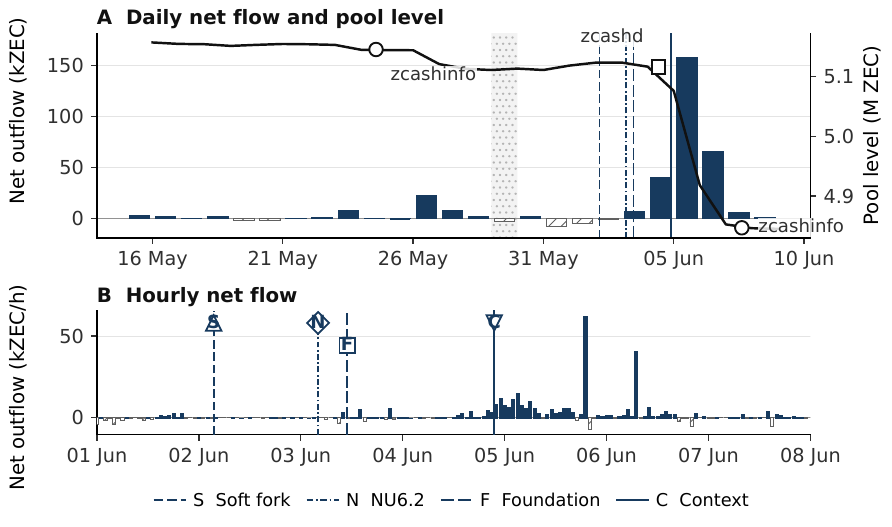}
\caption{Aggregate shielded-pool balance and net unshielding around the Orchard
disclosures. Panel A compares daily flow, reconstructed aggregate balance, and
independent snapshots; Panel B shows June 1--7
hourly flow. The dotted band marks discovery; navy and hatched bars denote net
outflow and inflow. Snapshot agreement supports the reconstruction, but net
movement is not total holder exit.}
\label{fig:pooltimeseries}
\end{figure}
\FloatBarrier

Among the paper's noncausal measures, this independently reconstructed flow is
the mechanism-nearest within-event directional test of P4, not an estimate of
structural $m$. Pause spending was disabled, the NU6.2/Foundation windows
overlap by 18 hours, and net flow can include delayed transfers or offsetting
inflows. Prices, sales, and flows are jointly determined, so the series does not
identify a causal ambiguity effect.

\subsection{Cross-Market Diagnostics}

Cross-market data ask whether the move is venue-specific or market-wide.
Under the 21:00 UTC hour-floor convention, Binance contract open interest rises
$21.8\%$ while its dollar value falls with
price and spot market-order flow remains net selling
\cite{binancefuturesapi,binancespotklines}. Relative to their pre-windows,
XMR/ETH weaken more, DASH improves,
and SOL is nearly flat \cite{coingeckoapi}. Thus the decline extends beyond one
venue but is not uniform across peers. Jointly determined prices, positions,
and flows do not identify its cause; Orchard's value is the staged comparison
of public evidence, price, and reconstructed flow.

\section{Cross-Event Evidence and the Identification Problem}\label{sec:panel}

\subsection{Applying \vdisc Beyond Orchard}

We apply Section~\ref{sec:data}'s rules to 14 disclosures with public records:
four official Bitcoin or Zcash cases
\cite{bitcoincore2018notice,electriccoin2019zcash,zcashsprout2026,wilcox2026orchard}
and ten bridge, contract, governance, oracle, or privacy proof failures. This
purposive sample tests measurement, not a population effect.

Daily CARs use an ETH/USDT market model estimated on days $-120$ to $-21$.
Calendar-day timing can include pre-disclosure movement: Orchard's 21:33 UTC
post would otherwise be credited with the whole June 4 return.
Table~\ref{tab:dose_response} therefore re-estimates the CAR--$z$ slope for the
nine priced events on the first three complete post-disclosure days. Re-timing
Orchard alone moves the slope from $-0.172$ to $-0.097$; applying the same rule
to all nine gives $+0.027$. The calendar-day interval excludes zero because
that clock credits Orchard's June 4 decline, which began before the post, to
the disclosure; both re-timed pairs-bootstrap 95\% intervals span zero, so the
pilot does not identify even the sign of the slope. Leverage, proxies, and
record completeness also differ.

\subsection{Comparison with Earlier Zcash Events}

Within Zcash, the aggregate shielded pool grows in the complete-proof case
(Sprout itself accepts no deposits \cite{zcashsprout2026}), while two unresolved cases
show outflow relative to baseline (Table~\ref{tab:dose_response}), consistent
with P4. Daily gross, daily net, and anchored net units differ, so response
magnitudes cannot be ranked.

\subsection{Why Returns Cannot Separate the Two Channels}

Among 38 events with a public loss figure (four official cases report zero),
market-cap-normalized reported loss and unresolved exposure co-move (Pearson
$0.86$, or $0.75$ without Multichain; Spearman $0.56$), so a negative return
cannot distinguish them. With all 38 events on the first three complete
post-disclosure days, residuals from a regression of the discount on loss/mcap
do not differ between the nine custody- or privacy-unverifiable events and the
rest (Welch $p=0.46$), and Orchard ranks 18/38. The panel documents confounding;
Orchard anchors and pool flows add timing/behavior evidence, not a separate
causal estimate.

The $0.86$ correlation is not an estimated verifiability effect, and part of it
is built into the records: in 17 of the 31 worksheet-coded events, the exposure
numerator equals the reported loss. A return regression then loads on both
channels unless external evidence measures proof and exposure separately; a
separate exposure field helps only when exposure is sourced independently of
the loss.

\input{tables/table_pool_dose_response}

The comparison shows why exposed state and public proof must be observable
before loss, exposure, and unresolved history collapse into one return.

\section{Protocol Design Implications}\label{sec:design}

Design asks whether outsiders can verify history ($q$), how much state one fault
can expose ($\theta$), and whether withdrawals reveal exit pressure ($m$).
Table~\ref{tab:design} maps each question to public evidence.

\input{tables/table4_design}

Replay, balance, and containment evidence strengthen $q$ or bound $\theta$;
boundary, migration, and disclosure records clarify operations and $m$. These
are complementary design candidates, not tested or causally ranked controls.

The distinction is operational. Replay and balance proofs tell outsiders what
can be ruled out about the past; containment and migration limit which state a
fault can affect. A residual-uncertainty notice does not close the unknown set,
but it prevents proved, unproved, and unknown claims from being conflated.

\section{Limitations}\label{sec:ethics}

The evidence ladder answers distinct questions: Orchard anchors order public
knowledge and returns; reconstructed flow checks holder movement; earlier Zcash
events test within-protocol direction; and the 40-event panel asks whether
returns separate loss from unresolved history. They do not in this sample.

The pattern fits the turnstile mechanism but is not a causal discount estimate:
the final statement bundled severity with unverifiability, price fell earlier,
event clocks and proxies differ, the 14 cases are purposive, and net flow is not
gross exit. Primary sources, fixed coding, and a frozen model-assisted-coding
audit make decisions reproducible without removing these limits.

This defensive study uses only public, patched-event evidence; it develops no
exploit code, targets no live system, and excludes unsafe raw outputs.

\section{Conclusion}

Privacy changes what repair can establish. A patch closes the disclosed flaw;
it does not prove whether private state was abused beforehand. We call the
market consequence of that gap the verifiability discount. It can matter even
with no reported theft and protected aggregate supply because public evidence
may still leave past abuse unresolved.

The turnstile model explains the incentive: an early exit can avoid part of a
possible shortfall and leave more to remaining holders. Orchard's staged
anchors match the predicted direction---a positive return after notice, near
zero after repair and supply assurance, and the largest negative post-anchor
return ($-0.162$) after past abuse could not be ruled out. Net unshielding over
the next 48 complete hours was $4.8\%$ of the aggregate shielded balance. These patterns, and earlier
Zcash comparisons, are directional evidence rather than causal estimates.

In 40 security events, market-cap-normalized reported loss and unresolved
exposure were correlated among the 38 with a public loss figure (Pearson
$0.86$, Spearman $0.56$), partly because exposure is often coded from the loss
itself. A return therefore cannot show whether markets priced reported
loss, unresolved history, or both.
Methodologically, \vdisc records source and time, loss, exposure, public proof,
return, and pool movement as separate fields. It then compares their timing
across Orchard, earlier Zcash cases, and the broader panel, making explicit what
outsiders can and cannot establish without treating an event return as causal.

Repair should provide evidence as well as corrected code. Protocols should
limit the state one flaw can expose, prepare evidence that preserves privacy
while bounding past abuse, and state what remains unknown. Resilience requires
both stopping future abuse and making pre-repair history as auditable as privacy
permits.

\paragraph{Declarations.}
No human-subject data are used; all inputs are public, timestamped sources. The
authors have no competing interests to declare. The views are those of the
authors, not their affiliated institutions, and this academic paper is not
investment advice. An AI agent assisted with language editing, table
generation, and consistency checks; the authors reviewed all output and take
full responsibility for the manuscript.

\FloatBarrier
\bibliographystyle{splncs04}
\bibliography{references}

\end{document}

%% file: tables/table_zec_case.tex
\begin{table}[ht]
\caption{ZEC cumulative abnormal returns before and after the four Orchard anchors.}
\label{tab:zecdiag}
\centering
\footnotesize
\setlength{\tabcolsep}{3.5pt}
\begin{tabularx}{\textwidth}{@{}
>{\raggedright\arraybackslash}p{0.14\textwidth}
>{\raggedright\arraybackslash}p{0.115\textwidth}
>{\raggedright\arraybackslash}X
>{\raggedleft\arraybackslash}p{0.08\textwidth}
>{\raggedleft\arraybackslash}p{0.08\textwidth}
>{\raggedright\arraybackslash}p{0.18\textwidth}@{}}
\toprule
Anchor & UTC & What became public & 24 h pre-CAR & 24 h post-CAR & Context-post tail test \\
\midrule
Public notice & 6/2 06:20 & Initial containment notice & 0.052 & 0.185 & \\
NU6.2 & 6/3 04:05 & Corrected circuit; no evidence of exploitation; supply unaffected & 0.185 & 0.014 & \\
Foundation post & 6/3 10:55 & Confirms no known exploitation and intact supply & 0.148 & 0.016 & \\
Context post & 6/4 21:33 & Prior exploitation judged unlikely but not ruled out & -0.134 & -0.162 &
$p_{ob}=0.003$; $q_{\mathrm{BH}}=0.051$ \\
\bottomrule
\end{tabularx}

\noindent\footnotesize\emph{Note.} CAR is the cumulative BTC-adjusted return.
Each pre- and post-window contains 24 complete hourly returns. For the context
post, $p_{ob}$ is the one-sided lower-tail placebo probability, the smallest
value attainable with 312 overlapping pre-discovery blocks, and
$q_{\mathrm{BH}}$ is the Benjamini--Hochberg adjusted value within the 16-test
family (4 anchors $\times$ 4 statistics).
\end{table}

%% file: tables/table_anchor_flows.tex
\begin{table}[ht]
\caption{Net movement out of the aggregate Sprout, Sapling, and Orchard shielded pools during the first 24 complete hours after the soft fork and each disclosure anchor.}
\label{tab:anchorflows}
\centering
\footnotesize
\begin{tabularx}{\textwidth}{@{}
>{\raggedright\arraybackslash}p{0.35\textwidth}
>{\raggedleft\arraybackslash}p{0.14\textwidth}
>{\raggedright\arraybackslash}X@{}}
\toprule
Anchor (pool state) & \multicolumn{1}{c}{Net unshielding (ZEC)} & How to read the comparison \\
\midrule
Soft-fork pause (Orchard off) & -7 & Constrained; not comparable \\
NU6.2 (restored) & 7,573 & Pent-up flow possible \\
Foundation post (operating) & 9,052 & Overlaps NU6.2 \\
Context post (operating) & 175,568 & About $49\times$ baseline \\
\bottomrule
\end{tabularx}

\noindent\footnotesize\emph{Note.} A positive value means that more ZEC left
the aggregate shielded pools than entered them; it is not the total amount withdrawn. The
pause window starts at 04:00, the first complete hour after the 03:30 activation
(a 02:00 start gives $-125$). The
NU6.2 and Foundation 24-hour windows overlap by 18 h. The May 15--28 baseline
is 3,560 ZEC per day.
\end{table}

%% file: tables/table_pool_dose_response.tex
\begin{table}[ht]
\caption{Cross-event directional comparisons and identification diagnostics.}
\label{tab:dose_response}
\centering
\footnotesize

\begin{tabularx}{\textwidth}{@{}
>{\raggedright\arraybackslash}p{0.17\textwidth} c
>{\raggedright\arraybackslash}p{0.16\textwidth}
>{\raggedright\arraybackslash}X
>{\raggedright\arraybackslash}p{0.17\textwidth}@{}}
\multicolumn{5}{@{}l}{\textbf{Panel A. Same-chain directional comparison}} \\
\toprule
Event & $z$ & Measure & Baseline $\to$ event & Within-event response \\
\midrule
Sprout 2026 & 0.000 & Daily aggregate net & 435 out/day $\to$ 14,026 inflow & Aggregate pool grows \\
Sprout 2019 & 0.020 & Daily gross & 8.8k/day $\to$ 25.4k out & $\approx2.9\times$ \\
Orchard 2026 & 0.162 & 24h aggregate net & 3,560 out/day $\to$ 175,568 out & $\approx49\times$ \\
\bottomrule
\end{tabularx}

\begin{tabularx}{\textwidth}{@{}
>{\raggedright\arraybackslash}p{0.25\textwidth}
>{\raggedright\arraybackslash}X
>{\raggedright\arraybackslash}p{0.18\textwidth}@{}}
\multicolumn{3}{@{}l}{\textbf{Panel B. Cross-event scope and frontier}} \\
\toprule
Exercise (coverage) & Result & Reading \\
\midrule
Clock sensitivity (14 events; 9 priced) & Orchard CAR $-0.334\to+0.015$ on the
first 3 complete post-days; CAR--$z$ slope $-0.172$ (calendar days),
$-0.097$ (Orchard re-timed), $+0.027$ (all nine re-timed) & Sign not
identified \\
Loss frontier (40 events; 38 with a loss figure) & $\mathrm{corr}(\mathrm{loss/mcap},z)=0.86$
(Spearman $0.56$); residual-group Welch $p=0.46$ & Loss--$z$ co-move, partly
by construction \\
Orchard residual (rank among 38) & $-0.57\sigma$; rank 18/38 & No outlier \\
\bottomrule
\end{tabularx}

\noindent\footnotesize\emph{Note.} Pairs-bootstrap 95\% intervals for the three
slopes are $[-2.237,-0.004]$, $[-2.103,0.282]$, and $[-0.402,0.306]$. Frontier
residuals use the first 3 complete post-days for all 38 events; re-timing
Orchard alone gives Welch $p=0.31$ and $-0.61\sigma$ (rank 19/38).
\end{table}

%% file: tables/table4_design.tex
\begin{table}[!ht]
\caption{Verifiability-by-design evidence map.}
\label{tab:design}
\centering
\small
\begin{tabularx}{\textwidth}{@{}
>{\raggedright\arraybackslash}p{0.18\textwidth}
>{\raggedright\arraybackslash}p{0.31\textwidth}
>{\raggedright\arraybackslash}X@{}}
\toprule
Mechanism & Uncertainty addressed & Public evidence produced \\
\midrule
\multicolumn{3}{@{}l}{\textit{Proof and accounting mechanisms}} \\
State replay & Historical non-abuse & Replayable logs, checkpoints, or snapshots \\
Balance proof & Aggregate supply & Aggregate-balance proof \\
Containment & Maximum affected set & Quarantine evidence or bounded loss sets \\
\midrule
\multicolumn{3}{@{}l}{\textit{Operational and disclosure controls}} \\
Turnstile gate & Cross-pool boundary & Boundary, conversion log, or invariant \\
Migration path & Future state; legacy-state treatment &
Migration/retirement proof; new keys \\
Residual notice & Unknowns explicit---not reduced &
Proved/unproved/unknown split \\
\bottomrule
\end{tabularx}
\end{table}